\documentclass[12pt]{article}%
\usepackage{amsfonts}
\usepackage{amssymb}
\usepackage{amsmath}
\usepackage{float}
\usepackage{color}
\usepackage{graphicx}
\usepackage[round,authoryear,comma]{natbib}
\usepackage{hyperref}%
\providecommand{\U}[1]{\protect\rule{.1in}{.1in}}
\restylefloat{figure}
\newtheorem{theorem}{Theorem}

\newtheorem{definition}{Definition}

\newtheorem{proposition}{Proposition}

\newtheorem{example}{Example}

\let\oldexample\example
\renewcommand{\example}{\oldexample\normalfont}
\newenvironment{proof}[1][Proof]{\noindent \textbf{#1:} }{\hfill \rule{0.5em}{0.5em}}

\parskip=\medskipamount
\begin{document}

\title{\textbf{Axiomatic characterizations of some simple risk-sharing rules}}
\author{Jan Dhaene \footnote{Corresponding author. Email: jan.dhaene@kuleuven.be}\\Actuarial Research Group, KU Leuven
\and Rodrigue Kazzi\\Actuarial Research Group, KU Leuven
\and Emiliano A. Valdez\\Department of Mathematics, Univ. of Connecticut}

\maketitle

\begin{abstract}

In this paper, we present axiomatic characterizations of some simple risk-sharing (RS) rules, such as the uniform, the mean-proportional and the covariance-based linear RS rules. These characterizations make it easier to understand the underlying principles when applying these rules. Such principles typically include maintaining some degree of anonymity regarding participants' data and/or incident-specific data, adopting non-punitive processes and ensuring the equitability and fairness of risk sharing. By formalizing key concepts such as the reshuffling property, the source-anonymous contributions property and the strongly aggregate contributions property, along with their generalizations, we develop a comprehensive framework that expresses these principles clearly and defines the relevant rules. To illustrate, we demonstrate that the uniform RS rule, a simple mechanism in which risks are shared equally, is the only RS rule that satisfies both the reshuffling property and the source-anonymous contributions property. This straightforward axiomatic characterization of the uniform RS rule serves as a foundation for exploring similar principles in two broad classes of risk-sharing rules, which we baptize the $q$-proportional RS rules and the $(q_1,q_2)$-based linear RS rules, respectively. This framework also allows us to introduce novel particular RS rules, such as the scenario-based RS rules.

\vspace{0.8cm}


\end{abstract}

\newpage

\section{Introduction} \label{sec:intro}

A risk-sharing arrangement is an agreement between different parties, such as individuals, families, business entities or governments, to share the potential costs or losses that might arise from uncertain events. These arrangements are structured to help reduce the burden of risk for any single participant by spreading it out across the group, thus making the impact of adverse outcomes for the individual participants more manageable. An example is ‘The Risk Sharing Platform,’ established in 2021. The goal of the Risk Sharing Framework developed by this platform is to provide a basis for interested humanitarian actors by improving risk management within delivery chains through a principled approach to sharing the burden of preventive measures and responsibility for materializing losses (\cite{risksharing2022}).

The concept of risk-sharing among participants has been in place for a long time, especially in traditional community agreements such as funeral cost-sharing, which exists in many parts of the world. As an example, consider the traditional `idir' collectives in Ethiopia. Members of an idir make monthly financial contributions, which are used to help participants to organize funerals for their relatives and to provide solace in grieving. Informal risk-sharing practices are still prevalent in many regions to help manage a variety of risks. For instance, small villages in India collaborate to mitigate risks associated with adverse weather affecting their crops (\cite{mazur2023sharing}), and in rural Tanzania, communities share the financial burdens of illness (\cite{de2006risk, de2011social}). Additionally, informal risk-sharing can manifest as social funds, which are community-managed resource pools that members can access in times of need (\cite{le2018interest, van2024autonomy}).

Recently, the concept of risk-sharing has received increased attention due to the growing importance of decentralization in developed economies. Unlike classical insurance models, where reliance is on a large insurance company, in decentralized systems, risks are managed directly among the participants (\cite{feng2023peer}). For example, `Broodfondsen', conceptualized in 2006 in the Netherlands and later imitated in the UK, are small-sized risk-sharing pools of self-employed individuals that provide mutual financial support in case of illness or disability (\cite{breadfunds2017, oostveen2018selfemployed}).

A wide variety of risks can be shared among community members. Risks can be broken down into systematic risks, which affect all participants in the same manner, and non-systematic (idiosyncratic) risks, which affect participants independently. Natural disasters, for instance, appear to be mainly systematic, but there is definitely an idiosyncratic component to them, as the damage could be local or the crops planted may be different. In rural areas, where unpredictable weather conditions like droughts, storms, or floods are common, farmers often establish mutual cooperatives. If such adverse weather damages their crops, the cooperative helps ensure that every farmer is equally protected, so that no one has to bear the full impact alone. This risk-sharing agreement helps farmers create a financial safety net that reduces the non-systematic component of the overall financial burden of such disasters. See \cite{vigani2019}.

One interesting aspect highlighted by \cite{risksharing2022} is the underlying principles deemed necessary for effective collaborative risk-sharing. These include maintaining the confidentiality of participants' data and incident-specific data, adopting non-punitive processes, and ensuring that risks are shared equitably rather than equally. This discussion sheds light on the fact that pools often prioritize specific principles of risk-sharing rather than specific rules. Thus, it becomes essential to understand what risk-sharing rule results from a set of requested principles. Conversely, when rules are chosen over principles, it is crucial to recognize the underlying principles that are implicitly endorsed when selecting a particular rule. This two-way question can be explored through axiomatizations and is the primary concern of our paper.

In order to formalize the principle of risk-sharing, consider $n$ economic agents, labeled $i = 1, 2, \ldots, n$. Starting with present time $0$, each agent $i$ faces a loss $X_i$ at time $1$. Without any form of insurance or pooling, agents would bear their own loss individually, i.e., at time $1$, agent $i$ would incur a loss equal to the realization of $X_i$.

All random variables (r.v.'s) considered in this paper are defined on a given probability space $(\Omega, \mathcal{F}, \mathbb{P})$. Equalities between r.v.'s are assumed to hold almost surely. A r.v. will always be denoted by an uppercase letter (e.g., $X_i$), with the subscript indicating the corresponding agent $i$. A random vector of size $n$ will be denoted by a bold uppercase letter, e.g., $\boldsymbol{X} = (X_1, X_2, \ldots, X_n)$.

Let $\chi$ be an appropriate set of r.v.'s on the probability space $(\Omega, \mathcal{F}, \mathbb{P})$. We interpret $\chi$ as the collection of losses (risks) of interest. The set $\chi$ is assumed to be a convex cone of r.v.'s on this probability space, which means that $\chi$ is closed under positive scalar multiplication and under addition. We assume that 0 is an element of  $\chi$. Depending on the situation at hand, the set $\chi$ could be defined as the set of all r.v.'s $X$ for which $\mathbb{E}[\left\vert X \right\vert] < \infty$, with $\mathbb{E}$ being the expectation under $\mathbb{P}$, commonly referred to as $L^1 = L^1(\Omega, \mathcal{F}, \mathbb{P})$, or the set $L^2 = L^2(\Omega, \mathcal{F}, \mathbb{P})$ of all r.v.'s with finite first and second moments. Other possible choices for $\chi$ include the set of all (essentially) bounded r.v.'s, denoted as $L^\infty = L^\infty(\Omega, \mathcal{F}, \mathbb{P})$, and the set of all r.v.'s, denoted as $L^0 = L^0(\Omega, \mathcal{F}, \mathbb{P})$. In addition, for any $L^p$ mentioned above, the subset $L_+^p$ consisting of all non-negative elements of $L^p$ might also be a suitable choice. In this paper, we will always silently assume that $\chi$ only contains non-negative r.v.'s (losses), although several results that we will present hereafter remain to hold (or can easily be adapted) when this non-negativity restriction is not made.  

The $n$-dimensional random vector of losses $\boldsymbol{X} = (X_1, X_2, \ldots, X_n)$ is referred to as the loss vector. The joint cumulative distribution function (cdf) of this loss vector $\boldsymbol{X}$ is represented by $F_{\boldsymbol{X}}$. The marginal cdfs of the individual losses are denoted by $F_{X_1}, F_{X_2}, \ldots, F_{X_n}$, respectively. The total or aggregate loss experienced by the $n$ agents with the loss vector $\boldsymbol{X}$ is given by
\begin{equation}
S_{\boldsymbol{X}} = \sum_{i=1}^{n} X_i. \label{eq1}
\end{equation}
Hereafter, we will frequently refer to $\boldsymbol{X}$ as the pool (of losses) and to each agent as a participant in the pool.

A reallocation of losses, as defined below, is a fundamental concept in a risk-sharing arrangement.

\begin{definition}
For any pool $\boldsymbol{X} $ with aggregate loss $S_{\boldsymbol{X}}$, the set $\mathcal{A}_{\boldsymbol{X}}$ is defined as
\[
\mathcal{A}_{\boldsymbol{X}} = \left\{ (Y_1, Y_2, \ldots, Y_n) \in (L^0)^n \,\bigg| \sum_{i=1}^{n} Y_i = S_{\boldsymbol{X}} \right\}.
\]
\end{definition}

The elements of $\mathcal{A}_{\boldsymbol{X}}$ are referred to as the $n$-dimensional reallocations of $\boldsymbol{X}$. Note that the set $\mathcal{A}_{\boldsymbol{X}}$, as defined in this article, is broader than the typical set of reallocations where $Y_i$'s are restricted to $\chi$, as in \cite{jiao2022axiomatic} or in \cite{denuit2024comonotonicity}. This intentional choice allows for greater flexibility. For instance, it permits elements of the reallocation to be negative even when the original losses are non-negative. 

Risk-sharing (often abbreviated as RS) within a pool $\boldsymbol{X} \in \chi^n$ is a two-stage process. In the \textit{ex-ante step} (at time 0), the losses $X_i$ within the pool are reallocated by transforming $\boldsymbol{X}$ into another random vector $\boldsymbol{C}[\boldsymbol{X}] \in \mathcal{A}_{\boldsymbol{X}}$:
\[
\boldsymbol{C}[\boldsymbol{X}] = \big(C_1[\boldsymbol{X}], C_2[\boldsymbol{X}], \ldots, C_n[\boldsymbol{X}]\big).
\]
In the \textit{ex-post step} (at time 1), each participant receives the realization of his initial loss $X_i$ from the pool and pays the realization of $C_i[\boldsymbol{X}]$ to the pool. Hereafter, $C_i[\boldsymbol{X}]$ is referred to as the contribution of participant $i$ to the pool, while $\boldsymbol{C}[\boldsymbol{X}]$ is called the contribution vector.

Since $\boldsymbol{C}[\boldsymbol{X}] \in \mathcal{A}_{\boldsymbol{X}}$, the risk-sharing process, as explained above, satisfies the \textbf{full allocation condition}:
\begin{equation}
\sum_{i=1}^{n}C_i[\boldsymbol{X}] = \sum_{i=1}^{n}X_i.
\label{eq2}
\end{equation}
This equality ensures that all losses are fully covered by the contributions, provided each participant is willing to pay his contribution at time 1. To decrease the probability of participant default, one could require each participant to pay a sufficiently high deposit at time 0. If this deposit exceeds the required contribution, the excess can be refunded to the participant at time 1. However, guaranteeing non-default by participants is only feasible if each deposit is equal to (or higher than) the maximum contribution that an individual might be required to make, which may be unreasonable to request.

In this paper, we will use the term decentralized insurance for a risk-sharing system where an administrator collects the contributions of the participants ex post and then redistributes these contributions to indemnify each participant for his loss. In the literature, decentralized insure is often called P2P insurance. However, it is important to note that the term P2P insurance is sometimes restricted to risk-sharing arrangements that involve direct transfers between participants, rather than relying on an administrator to first collect the contributions and then to redistribute them among participants in the pool. 

Although risk-sharing has a long history in actuarial science and insurance, recently several researchers have shown a renewed interest in this research topic in the decentralized setting. \cite{denuit2022risk} investigate the properties of several risk-sharing rules. One particular RS rule that they consider is the so-called conditional mean RS rule, which was introduced in the actuarial literature in \cite{denuit2012convex}. \cite{dhaene2024axiomatic} also introduce the quantile RS rule. \cite{denuit2021risk} study various other risk-sharing rules, including linear risk-sharing rule, that can be used in a large pool with heterogeneous losses including P2P setups. \cite{levantasi2022mutual} and \cite{clemente2023risk} explore P2P insurance arrangements that account for fluctuations in total losses by adding a safety margin to participant contributions. \cite{feng2023peer} propose a P2P risk-sharing model using convex programming to develop an optimal and fair way to share risks with a focus on applying it to a flood risk pooling. \cite{ghossoub2024pareto}
examine risk-sharing allocations that take into account the risk tolerance of each participant in relation to tail events. \cite{ghossoub2024efficiency} investigate Pareto efficiency in a pure-exchange economy where agents' preferences are represented by risk-averse monetary functions. As an application, they examine risk-sharing markets where all agents evaluate risk through law-invariant coherent risk measures. 

In this paper, we focus on developing basic principles, or axioms, that define certain simple RS rules such as the uniform, mean-proportional and covariance-based linear RS rules. We axiomatically characterize these rules and embed them within broader classes of proportional and linear RS rules, which we define and label as $q$-proportional and $(q_1,q_2)$-based linear RS rules. This approach also leads to the introduction of new, practical and simple specific RS rules such as the scenario-based RS rules. The axioms are based on principles commonly requested in pools, such as the degree of anonymity of participants' personal information or details on their individual losses. We will explicitly define several properties related to these aspects later in the paper. For example, we demonstrate how the uniform RS rule can be established by satisfying two fundamental requirements: the reshuffling property and the source-anonymous contributions property. Axiomatic approaches are crucial because they provide a clear framework for comparing and evaluating different decision rules, as seen in our case with various risk-sharing rules (\cite{gilboa2019axiomatizations}). Once a set of axioms is established, a RS rule that meets these criteria is considered consistent with the underlying principles.

There has been a significant amount of research on the axiomatic foundations of many decision-making rules including insurance pricing, e.g., \cite{wang1997pricing} and \cite{goovaerts1997prem}, but when it comes to risk-sharing rules, the research work is still quite limited. A few notable contributions include the work by \cite{denuit2022risk}, who consider an extensive list of properties that RS rules might satisfy. Additionally, \cite{jiao2022axiomatic} provide an axiomatic characterization of the conditional mean RS rule, and \cite{dhaene2024axiomatic} present an axiomatic characterization of the quantile RS rule. \cite{ghossoub2024allocation} propose several axiomatizations for robust versions of linear RS rules. Our work aims to further expand on this by developing axiomatic frameworks for some simple RS rules. We note that constructing RS rules using an axiomatic approach is just one way to define an appropriate rule based on the specific situation at hand. Alternatively, the choice of a RS rule can also be made by solving an optimization problem, the solution of which may lead to a suitable RS rule. A recent paper by \cite{yang2024optimality} offers interesting insights into this approach.

The rest of the paper is organized as follows. Section~\ref{sec:rsrules} introduces the fundamental concepts of risk-sharing rules and defines several traditional and new classes of these rules. In Section~\ref{sec:properties}, we introduce the properties that constitute the basis for the characterization of these rules.  Detailed discussions of two axiomatizations for the uniform RS rule,  the class of $q$-proportional (which includes the mean-proportional) RS rules and the class of $(q_1, q_2)$-based (which includes the covariance-based) linear RS rules are presented in Section~\ref{sec:axiom_uni}, Section~\ref{section:charac_qprop} and Section~\ref{section:charac_qqlinear}, respectively.  Finally, some concluding remarks and a summary of our findings are presented in Section 7. 

\section{Definition and examples of risk-sharing rules}\label{sec:rsrules}

We begin by formally defining what we mean by a risk-sharing rule, then introduce some classical and linear risk-sharing rules along with new `simple' classes of risk-sharing rules.

\begin{definition}
A risk-sharing rule for a given group of $n$ participants, each with losses in $\chi$, is a mapping $\boldsymbol{C} $  that transforms any pool $\boldsymbol{X}$ in $\chi^n$ into a contribution vector $\boldsymbol{C}[\boldsymbol{X}]$ in $\mathcal{A}_{\boldsymbol{X}}$:
\[
\boldsymbol{X} \in \chi^n \rightarrow \boldsymbol{C}[\boldsymbol{X}] \in \mathcal{A}_{\boldsymbol{X}}.
\]
\end{definition}

It is important to note that $\boldsymbol{C}$ is not simply a function from $\mathbb{R}^n$ to (a subset of) $\mathbb{R}^n$. Instead, it is a more general function from $\chi^n$ to (a subset of) $\ (L^0)^n$. Seen from time $0$, the contribution vector $\boldsymbol{C}[\boldsymbol{X}]$ is a random vector, as it depends on $\boldsymbol{X}$, and potentially on other sources of randomness. Introducing the notation $\sigma(\boldsymbol{X})$ for the $\sigma$-algebra generated by $\boldsymbol{X}$, this implies that $\boldsymbol{C}[\boldsymbol{X}]$ is not necessarily measurable with respect to $\sigma(\boldsymbol{X})$.

For a given group of participants, a risk-sharing rule is established and defined at time $0$ prior to the realization of any losses. Such a rule serves as a formal mechanism that determines how the total or aggregate loss, observable at time $1$, will be allocated among the participants or agents. The design and the implementation of the appropriate RS rule are critical for ensuring the success and the sustainability of the RS arrangement.

Hereafter, when considering a RS rule $\boldsymbol{C} $ and a pool $\boldsymbol{X}$, we will silently assume that $\boldsymbol{C}$ is defined in $\chi^n$ and that $\boldsymbol{X}$ is an element of $\chi^n$. In certain cases, it will be necessary to define $\chi^n$ more specifically, depending on the RS rule at hand.

\subsection{Two well-known risk-sharing rules}

Let us begin by introducing two simple but important RS rules: the uniform RS rule and the mean-proportional RS rule.

\begin{definition}[Uniform RS rule]
A RS rule $\boldsymbol{C}$ is called the uniform RS rule if, for any pool $\boldsymbol{X}$, the contribution vector is given by

\begin{equation} \label{eq:uniRS}
\boldsymbol{C}^{\mathrm{uni}}[\boldsymbol{X}] = \left(\frac{S_{\boldsymbol{X}}}{n}, \frac{S_{\boldsymbol{X}}}{n}, \ldots, \frac{S_{\boldsymbol{X}}}{n}\right). 
\end{equation}
\end{definition}

\begin{definition}[Mean-proportional RS rule]
A RS rule $\boldsymbol{C}$ in $\chi^n\subseteq (L_+^1)^n$ is the mean-proportional RS rule, if for any pool $\boldsymbol{X} $ with at least one $\mathbb{E}[X_j]>0$, the contribution vector is given by
\begin{equation} \label{eq:meanRS}
\boldsymbol{C}^{\mathrm{mean}}[\boldsymbol{X}] = \left(\frac{\mathbb{E}[X_1]}{\mathbb{E}[S_{\boldsymbol{X}}]} S_{\boldsymbol{X}}, \frac{\mathbb{E}[X_2]}{\mathbb{E}[S_{\boldsymbol{X}}]} S_{\boldsymbol{X}}, \ldots, \frac{\mathbb{E}[X_n]}{\mathbb{E}[S_{\boldsymbol{X}}]} S_{\boldsymbol{X}}\right).
\end{equation}
\end{definition}

The uniform RS rule distributes the aggregate loss $S_{\boldsymbol{X}}$ equally among all participants. It is the simplest non-trivial and most well-known RS rule. The mean-proportional RS rule allocates a participant-specific proportion of the aggregate loss $S_{\boldsymbol{X}}$ to all participants, based on the expected value of their individual losses within the pool. The uniform and the mean-proportional RS rules have been widely used in a risk-sharing context in practice.

To ensure that the mean-proportional RS rule is well-defined, it is assumed that all elements of $\chi$ have finite expectations. Also, notice that in case all expectations of the losses are zero, then all fractions in~\ref{eq:meanRS} are equal to $0/0$, and hence, the expression in~\ref{eq:meanRS} is not well-defined. Taking into account that we assumed that all losses are non-negative, this situation will only occur in the degenerate case that all losses are zero. In this case, we do not specify the appropriate reallocation, and we leave it at the administrator's discretion to decide what the particular reallocation is when this situation occurs. We will come back to this observation later in this paper when introducing the class of $q$-proportional RS rules, of which the mean-proportional RS rule is a particular element.  

\subsection{Linear risk-sharing rules}

The uniform RS rule treats all individual risks equally, regardless of their characteristics. In contrast, the mean-proportional rule allows for differentiation between risks based on their means. However, it may happen that participants want the risk-sharing rule to reflect more information about their individual risks. For example, they might want to include information on the variability of each risk and its dependence with the aggregate risk. This can be achieved through linear risk-sharing rules.

We recall two widely used versions of linear risk-sharing rules. These rules determine each participant $i$'s contribution by adjusting his mean loss $\mathbb{E}[X_i]$, based on how much he should be penalized or rewarded for the deviation of the total loss $S_{\boldsymbol{X}}$ from the total mean loss $\mathbb{E}[S_{\boldsymbol{X}}]$.

\begin{definition}[Covariance-based linear RS rule]\label{def:cov-linear}
A RS rule $\boldsymbol{C}$ in $\chi^n\subseteq (L_+^2)^n$ is called the covariance-based linear RS rule if for any pool $\boldsymbol{X}$ for which $\mathrm{var}(S_{\boldsymbol{X}})>0$, the contribution vector is defined as follows:
\[
\small
\boldsymbol{C}^{\mathrm{cov}}[\boldsymbol{X}] = \left(\mathbb{E}[X_1] + \frac{\mathrm{cov}(X_1, S_{\boldsymbol{X}})}{\mathrm{var}(S_{\boldsymbol{X}})} (S_{\boldsymbol{X}} - \mathbb{E}[S_{\boldsymbol{X}}]), \ldots, \mathbb{E}[X_n] + \frac{\mathrm{cov}(X_n, S_{\boldsymbol{X}})}{\mathrm{var}(S_{\boldsymbol{X}})} (S_{\boldsymbol{X}} - \mathbb{E}[S_{\boldsymbol{X}}]) \right).
\]
\end{definition}
Remark that the covariance-based RS rule is not specified in case the aggregate claims have a deterministic (constant) value. Further notice that although we assumed the losses to be non-negative, it may happen that the realization of the contribution is negative. 
\begin{definition}[Variance-based linear RS rule]\label{def:varRS}
A RS rule $\boldsymbol{C}$ in $\chi^n\subseteq (L_+^2)^n$ is called the variance-based linear RS rule if, for any pool $\boldsymbol{X}$ for which at least one $\mathrm{var}(X_j)>0$, the contribution vector is defined as follows:
\[
\small
\boldsymbol{C}^{\mathrm{var}}[\boldsymbol{X}] = \left( \mathbb{E}[X_1] + \frac{\mathrm{var}(X_1)}{\sum_{k=1}^n \mathrm{var}(X_k)} (S_{\boldsymbol{X}} - \mathbb{E}[S_{\boldsymbol{X}}]), \ldots, \mathbb{E}[X_n] + \frac{\mathrm{var}(X_n)}{\sum_{k=1}^n \mathrm{var}(X_k)} (S_{\boldsymbol{X}} - \mathbb{E}[S_{\boldsymbol{X}}]) \right).
\]
\end{definition}

The variance-based RS rule is not specified only in case all individual losses of the pool under consideration are deterministic. The standard deviation-based linear RS rule is defined by replacing the variance measure in Definition~\ref{def:varRS} with the standard deviation measure. In fact, there are many measures that can replace the mean and the variance, as will be discussed in Section~\ref{section:classes}.

In both the covariance-based and variance-based linear RS rules, the deviation $(S_{\boldsymbol{X}} - \mathbb{E}[S_{\boldsymbol{X}}])$ is distributed among the participants as adjustments over their individual mean losses. For the covariance-based linear RS rule, the adjustment of the contribution for participant $i$ is determined by the ratio of $\mathrm{cov}(X_i, S_{\boldsymbol{X}})$ to the total covariance $\sum_{i=1}^n \mathrm{cov}(X_i, S_{\boldsymbol{X}})$, which equals $\mathrm{var}(S_{\boldsymbol{X}})$. This ratio effectively modifies each participant's mean loss according to his contribution to the overall risk variability of the pool.


Conversely, in the variance-based linear RS rule, each participant $i$'s contribution deviates from his individual mean loss based on the ratio of $\mathrm{var}(X_i)$ to the total variance $\sum_{i=1}^n \mathrm{var}(X_i)$. This approach emphasizes the individual risks the participant brings to the pool, focusing on the inherent risk associated with each participant. Clearly, the variance-based linear RS rule treats any risk vector $\mathbf{X}$ as the covariance-based linear RS rule would treat a risk vector of independent copies of the individual risks in $\mathbf{X}$.

\subsection{Classes of risk-sharing rules}\label{section:classes}

For the uniform and mean-proportional RS rules, each participant's contribution is determined by a ratio of the total or aggregate loss, specifically $\frac{1}{n}$ and $\frac{\mathbb{E}[X_i]}{\mathbb{E}[S_{\boldsymbol{X}}]}$, respectively. We extend this approach by considering different types of measures to define the ratios, which leads us to introduce new risk-sharing rules.

In this paper, we define a (one-dimensional) risk metric $q$ as a function from $\chi$ to $\mathbb{R}_0^+$, the set of non-negative real numbers. Examples of risk metrics include $q[X]=1,  q[X] = \mathbb{E}[X]$ and $q[X] = X(\omega^\ast)$,
where $\omega^\ast$ is a scenario that is considered as `typical' in the world modeled by the probability space under consideration.

We now define the $q$-proportional RS rule.

\begin{definition}[$q$-proportional RS Rule]
Consider a risk metric $q : \chi \rightarrow \mathbb{R}_0^+$. A RS rule $\boldsymbol{C}$ is said to be the $q$-proportional RS rule if, for any pool $\boldsymbol{X} $ with at least one $q[X_j]>0$, the contribution vector is given by
\[
\boldsymbol{C}^{\mathrm{prop}}[\boldsymbol{X}] = \left(\frac{q[X_1]}{\sum_{k=1}^n q[X_k]} S_{\boldsymbol{X}}, \frac{q[X_2]}{\sum_{k=1}^n q[X_k]} S_{\boldsymbol{X}}, \ldots, \frac{q[X_n]}{\sum_{k=1}^n q[X_k]} S_{\boldsymbol{X}}\right).
\]
\end{definition}

Note that the $q$-proportional RS rule is not specified for those pools in $\chi^n$ for which all ${q[X_i]}$ are equal to zero. For instance, if the risk metric is the expectation, then the RS rule is not specified in case all losses are always equal to zero. We leave it at the discretion of the administrator, or the pooling community, to decide on how to cope with such a situation.  Note however that in case all ${q[X_i]}$ are positive and equal, the contribution vector $\boldsymbol{C}^{\text{prop}}[\boldsymbol{X}]$ leads to equal contributions for any participant. Taking into account this observation, it may be reasonable to define the $q$-proportional RS rule for all pools in $\chi^n$ by `reading' $0/0$ as $1/n$. 

Further note that by choosing the appropriate risk metrics, the $q$-proportional RS rule reduces to the uniform RS rule or the mean-proportional RS rule. Let us now present two other examples.

\begin{example}[Variance-proportional RS rule] \label{eq:varRS}
When contributions are based on the fluctuation of participants' losses, the variance-proportional RS rule may be used by setting $q[X_i] = \mathrm{var}(X_i)$ as the risk metric so that the contribution of participant $i$ is given by
\[
C_i[\boldsymbol{X}] = \frac{\mathrm{var}(X_i)}{\sum_{k=1}^n \mathrm{var}(X_k)} S_{\boldsymbol{X}} , \quad \text{for } i=1,2,\ldots,n.
\]
This approach leads to higher contributions for participants with more variable losses, compared to the contributions linked to more stable losses.
\end{example}

\begin{example}[Standard deviation-proportional RS rule] \label{eq:sdRS}
An alternative to the above approach is to use standard deviation to measure the volatility of participants' losses. Here, the risk metric is set as $q[X_i] = \sqrt{\mathrm{var}(X_i)} = \sigma(X_i)$ and the contribution formula is
\[
C_i[\boldsymbol{X}] = \frac{\sigma(X_i)}{\sum_{k=1}^n \sigma(X_k)} S_{\boldsymbol{X}} , \quad \text{for } i=1,2,\ldots,n.
\]
Similar to the variance proportional RS rule, this method imposes the highest contributions for the losses with the largest standard deviation.
\end{example}

Another RS rule that is $q$-proportional is based on choosing a preset scenario $\omega^* \in \Omega$, evaluating each individual risk at that scenario, and using it as the risk metric for distributing the losses among the participants. We baptize this rule the scenario-based RS rule. The advantage of this rule is that it does not require any knowledge on probability theory to be applied. One only needs `expert knowledge' and 'agreement' on what is considered as a 'typical scenario' and on the particular realizations of the losses in such a scenario. We define this RS rule as follows.

\begin{definition}[Scenario-based proportional RS rule]
A RS rule $\boldsymbol{C}$ is said to be a scenario-based linear RS rule if for any pool $\boldsymbol{X}$ with at least one $X_j(\omega^*)>0$, the contribution vector is defined by
\[
\small
\boldsymbol{C}^{\mathrm{scen,prop}}[\boldsymbol{X}] = 
\left(\frac{X_1(\omega^*)}{\sum_{k=1}^n X_k(\omega^*)} S_{\boldsymbol{X}}, \frac{X_2(\omega^*)}{\sum_{k=1}^n X_k(\omega^*)} S_{\boldsymbol{X}}, \ldots, \frac{X_n(\omega^*)}{\sum_{k=1}^n X_k(\omega^*)} S_{\boldsymbol{X}}\right),
\]
where \(\omega^*\) represents an a priori chosen possible (typical) state of the world.
\end{definition}


Inspired by covariance-based and variance-based linear RS rules, one might consider using (one- and two-dimensional) risk metrics different from the mean and covariances/variances to define other risk-sharing rules. Suppose these risk metrics can be expressed as functions $q_1: \chi \rightarrow   \mathbb{R}_0^+$ and $q_2: \chi^2 \rightarrow \mathbb{R}$. Based on these metrics, we can define the following class of linear risk-sharing rules.

\begin{definition}[$(q_1,q_2)$-based linear RS rule]
Consider the risk metrics $q_1 : \chi \rightarrow \mathbb{R}_0^+$ and $q_2 : \chi^2 \rightarrow \mathbb{R}$. A RS rule $\boldsymbol{C}$ is said to be the $(q_1,q_2)$-based linear RS rule if for any pool $\boldsymbol{X}$ for which the denominator in the following expression is not zero, the contribution vector is defined by
\[
\small
\boldsymbol{C}^{\mathrm{lin}}[\boldsymbol{X}] = \left(q_1[X_1] + \frac{q_2[X_1, S_{\boldsymbol{X}}]}{\sum_{k=1}^n q_2[X_k, S_{\boldsymbol{X}}]} (S_{\boldsymbol{X}} - \sum_{k=1}^n q_1[X_k]), \ldots, q_1[X_n] + \frac{q_2[X_n, S_{\boldsymbol{X}}]}{\sum_{k=1}^n q_2[X_k, S_{\boldsymbol{X}}]} (S_{\boldsymbol{X}} - \sum_{k=1}^n q_1[X_k])\right).
\]
\end{definition}

Examples of $q_1$ include $q_1[X_i] = \mathbb{E}[X_i]$ whereas examples of $q_2$ include $q_2[X_i, S] = \mathrm{cov}[X_i, S]$ and $q_2[X_i, S] = \mathrm{var}[X_i]$. Hence, both the covariance-based and the variance-based linear RS rules are specific elements of the general class $(q_1,q_2)$-based linear RS rules. Note that the $(q_1,q_2)$-based linear RS rule is not specified in case the denominator in the contribution vector is equal to $0$. A possible solution could be to set all fractions equal to $1/n$ in such a case.

One practical application of $(q_1,q_2)$-based linear risk-sharing  involves selecting specific scenario-based risk metrics. Suppose \(q_1[X_i] = X_i(\omega^*)\) and \(q_2[X_i, S] = (X_i(\overline{\omega}) - X_i(\underline{\omega}))(S(\overline{\omega}) - S(\underline{\omega}))\), where \(\omega^*\) represents a `typical' state of the world while \(\overline{\omega}\) and \(\underline{\omega}\) represent two `extreme'  states of the world. This choice of risk metrics avoids the need for knowledge of the distribution functions of the random losses and relies instead on their realizations in three specific scenarios: \(\omega^*\), \(\overline{\omega}\), and \(\underline{\omega}\). This approach leads us to define the scenario-based linear RS rule as follows:

\begin{definition}[Scenario-based linear RS rule]
A RS rule $\boldsymbol{C}$ is said to be a scenario-based linear RS rule if for any pool $\boldsymbol{X}$ with $S_{\boldsymbol{X}}(\overline{\omega}) \neq S_{\boldsymbol{X}}(\underline{\omega})$, the contribution vector is defined by
\[
\small
\boldsymbol{C}^{\mathrm{scen,lin}}[\boldsymbol{X}] = \left(X_1(\omega^*) + \frac{X_1(\overline{\omega}) - X_1(\underline{\omega})}{S_{\boldsymbol{X}}(\overline{\omega}) - S_{\boldsymbol{X}}(\underline{\omega})} (S_{\boldsymbol{X}} - S_{\boldsymbol{X}}(\omega^*)), \ldots, X_n(\omega^*) + \frac{X_n(\overline{\omega}) - X_n(\underline{\omega})}{S_{\boldsymbol{X}}(\overline{\omega}) - S_{\boldsymbol{X}}(\underline{\omega})} (S_{\boldsymbol{X}} - S_{\boldsymbol{X}}(\omega^*)) \right),
\]
where \(\omega^*\), \(\overline{\omega}\), and \(\underline{\omega}\) represent three a priori chosen possible (typical vs. extreme) states the world.
\end{definition}

We note that the introduced scenario-based proportional and linear risk-sharing rules can also be extended to allow for different states $\omega^*$, $\overline{\omega}$, and $\underline{\omega}$ for each participant, meaning that each individual risk is evaluated under different scenarios. However, we do not elaborate on this extension in this paper.

It is easy to verify that in case $ q_2[X_i, S_{\boldsymbol{X}}] = q_1[X_i] $ holds for any $i$ and any $\boldsymbol{X} \in \chi^n$, the $(q_1,q_2)$-based linear RS rule reduces to the $q$-proportional RS rule. Consequently, the $(q_1,q_2)$-based linear RS rule encompasses any RS rule discussed before in this paper as a particular case.

In the remainder of this paper, we will explore several characterizations of the RS rules presented so far. First, in the following section we define some properties that RS rules may (or may not) satisfy. 

\section{Properties of risk-sharing rules}\label{sec:properties}

As mentioned above, axiomatic characterizations of different classes of risk-sharing rules have been examined by \cite{jiao2022axiomatic} and by \cite{dhaene2024axiomatic}, who focus on conditional mean risk-sharing rules and quantile risk-sharing rules, respectively. Drawing inspirations from these two papers, we will identify appropriate properties (axioms) that characterize each of the RS rules listed in Section~\ref{sec:rsrules}, which represent some of the most commonly used RS rules in practice. In this section, we introduce some general properties that RS rules may or may not satisfy, which can be seen as possible candidates for the axioms characterizing these RS rules.

A reshuffle of the pool $\boldsymbol{X} = (X_{1}, X_{2}, \ldots, X_{n})$ is a random vector $\boldsymbol{X}^{\pi}$ defined by
\[
\boldsymbol{X}^{\pi} = (X_{\pi(1)}, X_{\pi(2)}, \ldots, X_{\pi(n)}),
\]
where $\pi = (\pi(1), \pi(2), \ldots, \pi(n))$ represents a permutation of the set $\{1, \ldots, n\}$. A reshuffle is also known as a permutation or a rearrangement. Note that $\boldsymbol{X}$ and $\boldsymbol{X}^{\pi}$ are composed of the same individual losses, with only their positions being rearranged. After the reshuffle, $X_{\pi(i)}$ becomes the new loss attributed to participant $i$. Obviously, if $\boldsymbol{X} \in \chi^{n}$, then $\boldsymbol{X}^{\pi} \in \chi^{n}$ as well.

\begin{definition}[Reshuffling]
A RS rule $\boldsymbol{C}$ satisfies the reshuffling property if for any $\boldsymbol{X} \in \chi^{n}$ and any of its reshuffles $\boldsymbol{X}^{\pi}$, the following holds:
\[
C_i[\boldsymbol{X}^{\pi}] = C_{\pi(i)}[\boldsymbol{X}], \quad \text{for any } i = 1, \ldots, n.
\]
\end{definition}

As an example, consider the pool $\boldsymbol{X} = (X_{1}, X_{2}, X_{3})$ and its reshuffled version $\boldsymbol{X}^{\pi} = (X_{3}, X_{1}, X_{2})$. Since $\pi(1) = 3$, the reshuffling property implies that
\[
C_{1}[\boldsymbol{X}^{\pi}] = C_{3}[\boldsymbol{X}].
\]
The reshuffling property is straightforward to interpret: it indicates that losses and contributions are interconnected such that when participants exchange their individual losses, their contributions are exchanged correspondingly. Loosely speaking, a RS rule satisfying the reshuffling property involves for each participant a contribution based on the loss he brings to the pool, but not on other characteristics of his identity. 

It is easy to verify that all the RS rules mentioned in Section~\ref{sec:rsrules} adhere to the reshuffling property. As an example of a RS rule that does not satisfy the reshuffling property, consider the order statistics RS rule $\boldsymbol{C}^{\mathrm{ord}}$ which is defined as
\begin{equation}\label{order_stat}
 \boldsymbol{C}^{\mathrm{ord}}[\boldsymbol{X}] = (X_{(1)}, X_{(2)}, \ldots, X_{(n)}),   
\end{equation}
where $X_{(i)}$ is the \(i\)-th order statistic, i.e., the \(i\)-th smallest value in $\boldsymbol{X}$. Thus, we have $X_{(1)} \leq X_{(2)} \leq \ldots \leq X_{(n)}$. This RS rule can be interpreted as ordering participants based on their ascending risk-bearing capacities, such as their wealth, with those having higher capacities contributing more. It is easy to verify that $\boldsymbol{C}^{\mathrm{ord}}$ does not satisfy the reshuffling property.

\begin{definition}[Source-anonymous contributions]\label{def:source_anonymous}
The contributions of a RS rule $\boldsymbol{C} $ are said to be source-anonymous if, for any pool  $\boldsymbol{X} $ and any of its reshuffles $\boldsymbol{X}^{\pi}$, it holds that
\[
C_{i}[\boldsymbol{X}^{\pi}] = C_{i}[\boldsymbol{X}] \quad \text{for any } i = 1, \ldots, n.
\]
\end{definition}

Source-anonymity of a risk-sharing rule means that the contributions are not tied to who specifically incurs the losses $X_{1}, X_{2}, \ldots, X_{n}$. Instead, each participant's contribution is based on the collective set of losses, without considering which agent experienced which loss. In other words, contributions are determined by the individual losses, but the source of these individual losses is irrelevant for determining these contributions. 

It can easily be verified that the uniform RS rule satisfies the source-anonymity condition, while the mean-proportional RS rule does not satisfy this property.

\cite{jiao2022axiomatic} introduced the concept of `risk anonymity' as a potential property of RS rules. This property was also considered in \cite{denuit2022risk}, who call a RS rule that satisfies the risk anonymity property an `aggregate' RS rule. An aggregate (or risk-anonymous) RS rule is such that for any pool $\boldsymbol{X}$, the contributions $\boldsymbol{C}[\boldsymbol{X}]$ are measurable with respect to $\sigma(S_{\boldsymbol{X}})$. Hereafter, we will often refer to an aggregate RS rule as a `RS rule with aggregate contributions'. This leads to the following definition. 
\begin{definition}[Aggregate contributions]
A RS rule $\boldsymbol{C} $ is said to have aggregate contributions if for any pool $\boldsymbol{X}$ there exists a function
\[
\mathbf{h} : \mathbb{R} \rightarrow \mathbb{R}^n
\]
such that the contributions of $\boldsymbol{X}$ are given by 
\[
C_i[\boldsymbol{X}] = h_i(S_{\boldsymbol{X}}) \quad \text{for any } i = 1, \ldots, n.
\]
\end{definition}

This aggregate property for a RS rule means that the randomness of the contributions is solely due to the randomness of the aggregate loss. The realizations of the individual losses $X_{i}$'s are not revealed by communicating the realizations of the contributions; only the realization of the aggregate risk $S_{\boldsymbol{X}}$ is involved. Simply stated, the only unknown (random) information at time 0 that influences how the aggregate loss is shared among participants at time 1 is the future realization of that aggregate loss. Note that this does not limit the use of any probabilistic information on $\boldsymbol{X}$, as for every pool $\boldsymbol{X}$, the function $\mathbf{h}$ can be defined based on the joint distribution function of $\boldsymbol{X}$. For instance, under the mean-proportional RS rule, the aggregate contributions are determined such that for every $\boldsymbol{X}$, the $h_i$'s are defined as $h_i(s)= \frac{\mathbb{E}[X_i]}{\mathbb{E}[S_{\boldsymbol{X}}]}s$, for any $s$.


\cite{borch1960} defines a ``non-olet'' RS scheme as one in which the contribution vector is measurable with respect to the aggregate loss of the pool under consideration. This concept is closely related to our definition of an aggregate RS rule: an aggregate RS rule leads to a 'non-olet' RS scheme for any pool. For more information on 'non-olet' schemes, we refer to  \cite{feng2023peer} and \cite{feng2024unified}. An explanation of why these schemes are called ``non-olet'' can be found in \cite{feng2023book}, which is the first actuarial book focusing on decentralized risk-sharing.


Participants may agree to disregard the probabilistic characteristics of individual risks and focus solely on the realized aggregate loss, distributing it among themselves without regard to individual risk profiles. To capture this concept, we define a specific subclass of aggregate RS rules, referred to as `strongly aggregate RS rules' or `RS rules with strongly aggregate contributions'.

\begin{definition}[Strongly aggregate contributions]\label{def:stong_agg}
A RS rule $\boldsymbol{C} $ is said to have strongly aggregate contributions if there exists a function
\[
\mathbf{h} : \mathbb{R} \rightarrow \mathbb{R}^n
\]
such that the contributions of any pool $\boldsymbol{X}$ are given by
\[
C_i[\boldsymbol{X}] = h_i(S_{\boldsymbol{X}}) \quad \text{for any } i = 1, \ldots, n.
\]
\end{definition}

The contributions of a strongly aggregate RS rule are also $\sigma(S_{\boldsymbol{X}})$-measurable, which means that any RS rule with strongly aggregate contributions qualifies as an aggregate RS rule. However, the inverse implication is not necessarily true. For instance, the uniform RS rule has strongly aggregate contributions, and hence is also an aggregate RS rule. In contrast, the mean-proportional RS rule is an aggregate RS rule, but there exists no function \(\mathbf{h} : \mathbb{R} \rightarrow \mathbb{R}^n\) such that $\boldsymbol{C}^{\text{mean}}[\boldsymbol{X}] = \mathbf{h}(S_{\boldsymbol{X}})$ for any $\boldsymbol{X} \in \chi^n$, which implies that the mean-proportional RS rule is not strongly aggregate.

\begin{proposition}\label{prop:agg-source}
If the contributions of a RS rule are strongly aggregate, then these contributions are also source-anonymous.
\end{proposition}

The proof of this proposition is straightforward. Note that the inverse implication of the proposition is not generally true. For example, consider the order statistics RS rule $\boldsymbol{C}^{\mathrm{ord}}$ defined above. It is easy to verify that $\boldsymbol{C}^{\mathrm{ord}}$ is source-anonymous but not aggregate and hence, also not strongly aggregate. The mean-proportional RS rule is an example of an aggregate RS rule that does not exhibit source-anonymous contributions. 

The conditional mean RS rule was introduced in \cite{denuit2012convex}. Under this rule, the contribution of each participant $i$ is defined as
\begin{equation}
    C_{i}\left[  \boldsymbol{X}\right] = \mathbb{E}[X_i | S_{\textbf{X}}].
\end{equation}
This rule provides an example of a RS rule that is aggregate but neither strongly aggregate nor source-anonymous. Note also that this rule is not a $q$-proportional RS rule as there exists no risk metric $q$ such that the conditional mean RS rule can be expressed as a $q$-proportional RS rule.

Another simple rule that may be useful for illustration is the all-in-one RS rule which is defined as
\[
\boldsymbol{C}[\boldsymbol{X}] = \left(S_{\boldsymbol{X}}, 0, 0, \ldots, 0\right),
\]
where only one participant, say the first economic agent, is fully responsible for the total loss. Indeed, this RS rule has source-anonymous and strongly aggregate contributions but does not satisfy the reshuffling property since we are determining the contributions of the participants  based on their position in the risk vector (which could reflect their identity) rather than on their individual risks. During the COVID-19 pandemic, many governments around the world implemented large-scale economic relief programs that applied risk-sharing similar to an all-in-one RS rule. The government took on the role of the primary economic agent responsible for providing financial supports to struggling sectors and the general public health needs and safety to mitigate a devastating economic fallout.

We summarize the discussion of the properties satisfied by some RS rules in Table~\ref{table:properties}. A checkmark (\checkmark) means 'property is always satisfied,' while a dash (-) means 'property is not always satisfied'. We note that the mean-proportional RS and the scenario-based proportional RS rules defined in the previous section belong to the class of $q$-proportional RS rules. Similarly, the covariance-based linear RS and the scenario-based linear RS rules fall under the class of $(q_1, q_2)$-based linear RS rules. The members of these two classes always satisfy the properties of reshuffling and aggregate contributions, but they do not necessary satisfy the properties of source-anonymous contributions and strongly aggregate contributions. We will modify these latter properties as necessary in order to fully characterize these respective classes.

\begin{table}[htbp]
\caption{Properties of some risk-sharing rules}\label{table:properties}
\centering
\scalebox{0.95}{
\begin{tabular}{|l|c|c|c|c|} 
\hline
              &    & Source-  & & Strongly \\ 
Risk-sharing (RS)  &    & anonymous & Aggregate & aggregate \\
rules  & Reshuffling & contributions & contributions & contributions \\
\hline
Order statistics RS & $-$  & \checkmark  & $-$  & $-$  \\
\hline
Conditional mean RS & \checkmark  & $-$  & \checkmark  & $-$  \\
\hline
Mean-proportional RS & \checkmark  & $-$\  & \checkmark  & $-$  \\
\hline
Scenario-based proportional RS & \checkmark  & $-$\  & \checkmark  & $-$  \\
\hline
Scenario-based linear RS & \checkmark  & $-$\  & \checkmark  & $-$  \\
\hline
All-in-one RS & $-$  & \checkmark  & \checkmark  & \checkmark  \\
\hline
Uniform RS  & \checkmark  & \checkmark  & \checkmark  & \checkmark  \\
\hline
\end{tabular}}
\end{table}

\section{Characterizing the uniform RS rule}\label{sec:axiom_uni}

In this section, we present two straightforward (axiomatic) characterizations of the uniform RS rule. We start by demonstrating that $\boldsymbol{C}^{\mathrm{uni}}$ can
be characterized as the RS rule satisfying the `reshuffling property' and
having `source-anonymous contributions'.

\begin{theorem}\label{thrm:axiom-uni-1}
Risk-sharing rule $\boldsymbol{C}$ is the uniform RS rule if and only if it satisfies the following two axioms:
\begin{enumerate}
    \item[1:] $\boldsymbol{C}$ satisfies the reshuffling property.
   
     \item[2:] $\boldsymbol{C}$ has source-anonymous contributions.
\end{enumerate}
\end{theorem}

\begin{proof}
The ``only if'' statement follows from Table~\ref{table:properties}. Let us now prove the ``if'' statement. Consider the pool $\boldsymbol{X}\in\chi^{n}$ and its reshuffle $\boldsymbol{X}^{\pi}$. Since the contributions under $\boldsymbol{C}$
are assumed to be source-anonymous, we find that
\[
C_{i}\left[  \boldsymbol{X}\right]  =C_{i}\left[  \boldsymbol{X}^{\pi}\right]
\text{, \quad for any }i=1,\ldots,n.
\]
Taking into account the reshuffling property, these expressions can be
rewritten as
\[
C_{i}\left[  \boldsymbol{X}\right]  =C_{\pi(i)}\left[  \boldsymbol{X}\right]
\text{, \quad for any }i=1,\ldots,n.
\]
Since these expressions hold for any permutation $\pi$, it follows that all
contributions must be equal:
\[
C_{1}\left[  \boldsymbol{X}\right]  =C_{2}\left[  \boldsymbol{X}\right]
=\ldots=C_{n}\left[  \boldsymbol{X}\right].
\]
We arrive at the desired result by taking into account the full allocation condition (\ref{eq2}). 
\end{proof}

\bigskip

Theorem~\ref{thrm:axiom-uni-1} demonstrates that the uniform RS can be fully characterized by the reshuffling and source-anonymous contributions properties. This rule only requires the realization of the aggregate loss to determine each participant's contribution. Under conditions 1 and 2 in Theorem~\ref{thrm:axiom-uni-1}, there is no need for information about the individual losses or the multivariate distribution of the pool $\boldsymbol{X}$ to determine the contributions.

The characterization of the uniform RS rule in Theorem~\ref{thrm:axiom-uni-1} is somewhat trivial. Indeed, if reshuffling the losses reshuffles the contributions in the same way, while this reshuffling does not change the initial contributions, then all contributions must be equal to the contributions of the uniform RS rule, due to the full allocation condition. Further in the paper, the characterization in Theorem~\ref{thrm:axiom-uni-1} will serve as a `source of inspiration' for characterizing other RS rules. Hereafter, we will consider several adaptations of `source anonymity' that will lead to other RS rules.  

The elements of a set of properties (or axioms) are said to be 'independent' if not any of the properties contained in this set follows from the other properties in that set. Due to this 'non-redundancy', a set of independent axioms is sometimes called a `non-redundant' set of axioms.

\begin{proposition}\label{prop:theorem1}
The axioms 1 and 2 considered in Theorem~\ref{thrm:axiom-uni-1} are independent.
\end{proposition}

\begin{proof}
To prove the proposition, we need to show that for each of the two properties, there exists a RS rule that satisfies this property but not the other. Consider the stand-alone RS rule $\boldsymbol{C}$ defined by
\begin{equation}
\boldsymbol{C}[\boldsymbol{X}] = \boldsymbol{X}, \quad \text{ for any pool } \boldsymbol{X}. \label{UR1}
\end{equation}
This rule satisfies the reshuffling property but lacks source-anonymous contributions. Conversely, the order statistics RS rule $\boldsymbol{C}^{\text{ord}}$ has source-anonymous contributions but does not satisfy the reshuffling property.
\end{proof}

Taking into account that the uniform RS rule satisfies axioms 1 and 2, the independence of these axioms implies that neither of them alone is sufficient to characterize the uniform RS rule. 

Let us now prove a second characterization of the uniform RS rule.

\begin{theorem}\label{thrm:axiom-uni-2}
Risk-sharing rule $\boldsymbol{C}$ is the uniform RS rule if and only if the following two axioms hold:
\begin{enumerate}
    \item[1:] $\boldsymbol{C}$ satisfies the reshuffling property.
    \item[3:] $\boldsymbol{C}$ has strongly aggregate contributions.
\end{enumerate}
\end{theorem}

\begin{proof}
The ``if'' statement follows directly from Proposition~\ref{prop:agg-source} and Theorem~\ref{thrm:axiom-uni-1}. Conversely, the ``only if'' statement follows from Table~\ref{table:properties}.
\end{proof}

Let us now prove the independence of the two axioms considered in Theorem~\ref{thrm:axiom-uni-2}.

\begin{proposition}
The axioms 1 and 3 considered in Theorem~\ref{thrm:axiom-uni-2} are independent.
\end{proposition}

\begin{proof}
The stand-alone RS rule $\boldsymbol{C} $, as defined in (\ref{UR1}), satisfies the reshuffling property but is not strongly aggregate.
The all-in-one RS rule defined by
\[
\boldsymbol{C}[\boldsymbol{X}] = \left(S_{\boldsymbol{X}}, 0, 0, \ldots, 0\right),
\]
is strongly aggregate, while it violates the reshuffling property.
\end{proof}

From Theorem 2 and Proposition 3, we can conclude that neither of axioms 1 and 3 alone is sufficient to characterize the uniform RS rule. 

From Proposition~\ref{prop:agg-source}, we know that a strongly aggregate RS rule always exhibits source-anonymous contributions. This implies that axiom 2 in Theorem~\ref{thrm:axiom-uni-1} is a weaker requirement than axiom 3 in Theorem~\ref{thrm:axiom-uni-2}. In other words, there are more RS rules that satisfy property 2 than property 3. As a result, the ``only if'' statement (i.e., the `$\Rightarrow$' implication) of Theorem~\ref{thrm:axiom-uni-2} is stronger than the corresponding implication in Theorem~\ref{thrm:axiom-uni-1}. On the other hand, the more important ``if'' statement (i.e., the `$\Leftarrow$' implication) is stronger in Theorem~\ref{thrm:axiom-uni-1} than in Theorem~\ref{thrm:axiom-uni-2}. 

In Theorem~\ref{thrm:axiom-uni-2}, we have shown that a strongly aggregate RS rule satisfying the reshuffling property must be the uniform RS rule. However, notice that in Theorem~\ref{thrm:axiom-uni-2}, we cannot replace the strongly aggregate property with the weaker property that the RS rule is an aggregate RS rule. A simple example illustrating this remark is the mean-proportional RS rule, which is an aggregate RS rule and satisfies the reshuffling property, yet clearly differs from the uniform RS rule. 

\section{Characterizing $q$-proportional RS rules}\label{section:charac_qprop}

In this section, we explore two characterizations of $q$-proportional RS rules. These characterizations are derived by modifying the axioms presented in Theorem~\ref{thrm:axiom-uni-1} and Theorem~\ref{thrm:axiom-uni-2}, which led to the uniform RS rule.

We know that $q$-proportional RS rules always satisfy the reshuffling property but do not always meet the source-anonymous contributions property. This means that not all $q$-proportional RS rules satisfy the condition that $C_{i}\left[\boldsymbol{X}^{\pi}\right]  =C_{i}\left[\boldsymbol{X}\right]  $ for any participant $i$. To address this issue, we will modify the second property by defining a similar yet more suitable property for $q$-proportional RS rules. Therefore, for each participant $i$ and for each risk metric $q$, we first introduce the `contribution-over-$q$ ratio' $\frac{C_{i}\left[\boldsymbol{X}\right]
}{q\left[ X_{i}\right] }$, which measures the contribution of each participant $i$ in units of $q\left[ X_{i}\right]$. E.g., in case the risk metric $q$ is the expectation of the individual loss under consideration, then the contribution-over-$q$ ratio expresses the contribution of each participant in units of its expected loss. 

Consider, for example, the pool $(X,Y)$ consisting of two losses $X$ and $Y$ with different expectations. Requiring source-anonymous contributions, that is the contributions  $C_i[(Y,X)]$ and $C_i[(X,Y)]$ for both participants are equal, may not be appropriate or meaningful. The contribution-over-expectation ratio enables us to compare the relative sizes of contributions for different losses with varying expectations and it might be more appropriate to require that $C_1[(Y,X)]= \mathbb{E}[Y] / \mathbb{E}[X]  \times  C_1[(X,Y)]$ and $C_2[(Y,X)]= \mathbb{E}[X] / \mathbb{E}[Y]  \times  C_1[(X,Y)]$. This requirement means that in case of a reshuffle of the losses, only the relative change of the expectations of these losses is used  for determining the new contributions after the reshuffle. In other words, if in the original pool $(X,Y)$ the contribution of a participant is 2 times the expectation of his loss $X$, then the contribution of this participant in the reshuffled pool $(Y,X)$ is also twice the expectation of his newly attributed loss $Y$. In this simple example, we illustrated the use of the expectation as a risk metric to require how  contributions could be adapted after a reshuffle, but it is obvious that any other risk metric could be used as well in this perspective.

Let us now consider RS rules where we do not require that the contributions of
different participants remain unchanged when the individual losses are reshuffled. Instead, we require that the contribution ratios of the different participants remain unaffected by reshuffling the losses. 

\begin{definition}
[Source-anonymous contribution-over-$q$ ratios]\label{def:source_anonymous_ratios} Consider the risk metric $q : \chi \rightarrow \mathbb{R}_0^+$. The contribution-over-$q$ ratios of a RS rule 
$\boldsymbol{C}$  are said to be source-anonymous if, for any pool $\boldsymbol{X}$ and any of its reshufflings $\boldsymbol{X}^{\pi}$, the following conditions hold:
\[
{C_{i}\left[  \boldsymbol{X}^{\pi}\right]  }=\frac{q\left[  X_{\pi
(i)}\right]  }{q\left[
X_{i}\right]  }{C_{i}\left[  \boldsymbol{X}\right]  }\text{\quad for any }i=1,\ldots,n \text{\quad with } {q\left[  X_{i}\right] >0 }.
\]
\end{definition}

Remark that in case $q$ is a positive constant mapping, the source-anonymous contribution-over-$q$ ratios property reduces to the source-anonymous contributions property. 


When a RS rule exhibits source-anonymous contribution ratios, the specific allocations of the losses $X_{1},X_{2},\ldots,X_{n}$ to individual participants are irrelevant for determining their contribution ratios. In other words, the contribution ratios are independent of which participant incurs which specific loss. The determination of the contribution ratios depends solely on the collection of risks, not on the identity of the participant bearing the loss. For a RS rule with source-anonymous contribution ratios, only the risk metrics of the losses are required to adjust the original contributions after a reshuffle. For instance, if a participant's risk metric doubles after a reshuffle, then his contribution will also be doubled. 

It is easy to verify that any $q$-proportional RS rule has source-anonymous contribution-over-$q$ ratios. In particular, we find that both the uniform RS rule and the mean-proportional RS rules satisfy this property.

\begin{theorem}\label{thrm:axiom-prop-1}
Consider the risk metric $q : \chi \rightarrow \mathbb{R}_0^+$. A RS rule $\boldsymbol{C}$  is the $q$-proportional RS rule if and only if it satisfies the following two axioms:
\begin{enumerate}
    \item[1.] $\boldsymbol{C}$ satisfies the reshuffling property.
    \item[4.] $\boldsymbol{C}$ has source-anonymous contribution-over-$q$ ratios.
\end{enumerate}
\end{theorem}

\begin{proof}
The ``only if'' direction is straightforward. 
Let us now prove the ``if''
direction. 

Consider a pool $\boldsymbol{X}\in\chi^{n}$ and its
reshuffle $\boldsymbol{X}^{\pi}$. 
As the $q$-proportional RS rule is only defined in case not all  ${q\left[  X_{i}\right]}$ are equal to $0$, let us suppose that ${q\left[  X_{j}\right]>0}$. Given that the contribution ratios are assumed to be source-anonymous, we have that
\[
{C_{j}\left[  \boldsymbol{X}^{\pi}\right]  }=\frac{q\left[  X_{\pi
(j)}\right]  }{q\left[
X_{j}\right]  }{C_{j}\left[  \boldsymbol{X}\right]  }.
\]
Taking into account the reshuffling property, this leads to
\[
C_{\pi(j)}\left[\boldsymbol{X}\right] =\frac{q\left[  X_{\pi
(j)}\right]  }{q\left[
X_{j}\right]  }{C_{j}\left[  \boldsymbol{X}\right]  }.
\]
Since this expression holds for any permutation $\pi$, we have that 
\[
C_{k}\left[\boldsymbol{X}\right] =\frac{q\left[  X_{k}\right]  }{q\left[
X_{j}\right]  }{C_{j}\left[  \boldsymbol{X}\right]  }
\]
holds for any $k$. Hence, there exists a r.v. $Y$ such that 

\[
{C_{k}\left[  \boldsymbol{X}\right]  }={q\left[  X_{k}\right]
\times Y}\text{, \quad for any }k=1,\ldots,n.
\]
Taking into account
the full allocation condition (\ref{eq2}), we find that:
\[
Y=\frac{S_{\boldsymbol{X}}}{\sum_{k=1}^n q[X_k] },
\]
which shows that the contributions are equal to the contributions of the $q$-proportional RS rule. 
\end{proof}

In the following proposition, we show that the axioms considered in Theorem~\ref{thrm:axiom-prop-1} to characterize $q$-proportional RS rules are independent.

\begin{proposition}\label{prop:theorem3}
The axioms 1 and 4 considered in Theorem~\ref{thrm:axiom-prop-1} are independent.
\end{proposition}

\begin{proof}
The stand-alone RS rule defined in \eqref{order_stat} is an example of a RS rule satisfying the reshuffling axiom, but in general not having source-anonymous contribution-over-$q$ ratios. \\
Next, consider the following RS rule
$\boldsymbol{C}$ in $\chi^{n}$:
\[
\boldsymbol{C}\left[  \boldsymbol{X}\right]  =\left\{
\begin{array}
[c]{cc}%
 (X_{(1)}, X_{(2)}, \ldots, X_{(n)}), &
:\text{ if }q[X_{1}]=q[X_{2}]=\ldots=q[X_{n}]\\ \left(\frac{q[X_1]}{\sum_{k=1}^n q[X_k]} S_{\boldsymbol{X}}, \frac{q[X_2]}{\sum_{k=1}^n q[X_k]} S_{\boldsymbol{X}}, \ldots, \frac{q[X_n]}{\sum_{k=1}^n q[X_k]} S_{\boldsymbol{X}}\right)& :\text{ otherwise, }
\end{array}
\right.
\]
where $X_{(i)}$ is the \(i\)-th order statistic of the pool under consideration. This RS rule has source-anonymous contribution-over-$q$ ratios and thus satisfies axiom 4, but it does not satisfy the reshuffling axiom 1.
\end{proof}

In the following example, we present a direct extension of the $q$-proportional RS rule which accounts for the level of reliability of the data provided by each participant $i$ to determine the value of his risk metric $q[X_i]$. 

\begin{example} [Weighted $q$-proportional RS rule]
    Let $w_1, \dots, w_n$ be positive real numbers. The weighted $q$-proportional RS rule  is such that the contribution of participant $i$ is given by
\begin{equation}\label{eq:w_q_RS}
    C_i[\mathbf{X}] = \frac{w_i q[X_i]}{\sum_{k=1}^{n}w_k q[X_k]} S_{\mathbf{X}}.
\end{equation}
With this RS rule, the contribution to be paid by each participant $i$ is not only based on the aggregate claims, but also on all the participants' identities (through the $w_j$'s) and their random losses (through the $q[X_j]$'s). The $w_j$'s can be interpreted as participant-specific corrections applied to the risk metrics, the $q[X_j]$'s. For instance, the data quality provided by participant $j$, whether good or poor, might result in the same estimation of $q[X_j]$, but it could influence its trustworthiness. In such a situation, the correction factor $w_j$ may be used as an adjustment factor to account for the model uncertainty in the calculation of $q[X_j]$. Suppose, for example, that participant $i$ is delivering poorer quality data. This will then typically result in a higher $w_i$  and thus lead to a higher contribution $C_i[\mathbf{X}]$. 
It is clear that if $w_1, \dots, w_n$ are not all equal, then the weighted $q$-proportional RS rule satisfies neither the reshuffling axiom 1 nor the source-anonymous contribution-over-$q$ ratios axiom 4, implying that this RS rule does not belong to the general class of $q$-proportional RS rules. 
\end{example}

In the previous section, we introduced a second characterization of the
uniform RS rule based on the reshuffling property and the strongly aggregate
contributions property. It is evident that the latter property does not apply to the
$q$-proportional RS rule when $q$ is a non-constant mapping. 
Hereafter, we introduce the strongly aggregate contribution-over-$q$ ratios property, which is an adaptation of the strongly aggregate contributions property, modified to align with the $q$-proportional RS rule.
Notice that a risk metric  $q : \chi \rightarrow \mathbb{R}_0^+$ is said to be normalized if $q[0]=0$ and additive if
$
q\left[\sum_{k=1}^{n}X_{k}\right]=\sum_{k=1}^{n}q[X_{k}]\text{,  for any
}\boldsymbol{X}\in\chi^{n}$. 
Examples of normalized and additive risk metrics are
\(
q[X]=\mathbb{E}\left[  X\right]
\)
and 
\(
q[X]=X\left(  \omega^{\ast}\right)   
\) for any $X\in\chi$, and
where $\omega^{\ast}$ is a `typical' scenario.

\begin{definition}\label{def:strong_agg_ratio}
[Strongly aggregate contribution-over-$q$ ratios] Consider the normalized and additive risk metric $q : \chi \rightarrow \mathbb{R}_0^+$. A RS rule $\boldsymbol{C}$ is said to have strongly aggregate contribution-over-$q$ ratios if there exists a
function
\[
\mathbf{h}: \mathbb{R}^{2} \rightarrow \mathbb{R}^{n}
\]
such that for any $\boldsymbol{X}$ with at least one ${q\left[  X_{j}\right] >0 }$, the contributions can be  expressed as:
\[
C_{i}\left[  \boldsymbol{X}\right]={q\left[  X_{i}\right]}\times h_{i}\left(  S_{\boldsymbol{X}},q\left[  S_{\boldsymbol{X}}\right]
\right)  \text{,\quad for any }i=1,\ldots,n.
\]
\end{definition}

It follows immediately that for any additive risk metric $q$, the $q$-proportional RS rule has strongly aggregate contribution-over-$q$ ratios. To illustrate, by choosing the expectation as the risk metric $q$, we can write the contribution formula corresponding to the mean-proportional RS rule as follows:
\begin{align*} 
C_i^\mathrm{mean}[\mathbf{X}] & = \mathbb{E}[X_i] \cdot \frac{S_{\mathbf{X}}}{\mathbb{E}[S_{\mathbf{X}}]} \\
& =\mathbb{E}[X_i] \cdot h_i(S_{\mathbf{X}}, \mathbb{E}[S_{\mathbf{X}}]), \quad \text{for any } i=1,\ldots,n.
\end{align*}
This formula highlights that the contribution of participant $i$  is proportional to the expected loss $\mathbb{E}[X_i]$, adjusted by a factor $h_i(S_{\mathbf{X}}, \mathbb{E}[S_{\mathbf{X}}])$, the proportion of the aggregate loss $S_{\mathbf{X}}$ to its expectation $\mathbb{E}[S_{\mathbf{X}}]$. This ensures that contributions under this rule are suitably scaled relative to the group's overall risk exposure.


Now we can prove a second characterization of the $q$-proportional RS rule when $q$ is additive.

\begin{theorem}\label{thrm:axiom-prop-2}
Consider the normalized and additive risk metric $q : \chi \rightarrow \mathbb{R}_0^+$. A RS rule $\boldsymbol{C}$ is the
$q$-proportional RS rule if and only if it satisfies the following
axiom:
\begin{enumerate} 
\item[5.]
$\boldsymbol{C}$ has strongly aggregate contribution-over-$q$ ratios.\
\end{enumerate}

\end{theorem}

\begin{proof}
The ``only if'' direction of the proof is straightforward. \\
Now, let us prove the ``if'' statement. 
Suppose that $\boldsymbol{C}$ has strongly aggregate contribution-over-$q$ ratios. Consider the pool $\boldsymbol{X}$ and assume that at least one ${q\left[  X_{j}\right] >0 }$. The contribution of participant $1$ can then be expressed as follows:
\begin{align} 
C_{1}\left[  \boldsymbol{X}\right]   &  =q\left[  X_{1}\right]
\times h_{1}\left(  S_{\boldsymbol{X}},q\left[  S_{\boldsymbol{X}%
}\right]  \right) \nonumber\\
&  =\frac{q\left[  X_{1}\right]  }{q\left[
S_{\boldsymbol{X}}\right]  }\times C_{1}\left[  \left(  S_{\boldsymbol{X}%
},0,\ldots,0\right)  \right]  . \label{Q1}%
\end{align}
On the other hand, for $i=2,3,\ldots,n,$ we have that
\[
C_{i}\left[  \left(  S_{\boldsymbol{X}},0,\ldots,0\right)  \right]  =0\text{.}%
\]
Thus, by the full allocation condition, we obtain
\[
C_{1}\left[  \left(  S_{\boldsymbol{X}},0,\ldots,0\right)  \right]
=\sum_{i=1}^{n}C_{i}\left[  \left(  S_{\boldsymbol{X}},0,\ldots,0\right)
\right]  =S_{\boldsymbol{X}}\text{.}%
\]
From (\ref{Q1}), we can conclude that
\[
C_{1}\left[  \boldsymbol{X}\right]  =\frac{q\left[  X_{1}\right]
}{\sum_{k=1}^n q[X_k]}\times S_{\boldsymbol{X}}.
\]
By applying a similar argument for participants $2,3,\ldots,n$, we can conclude that $\boldsymbol{C}$ is equal to the contribution vector of the $q$-proportional RS rule. This reasoning holds for any pool  $\boldsymbol{X}$ with at least one ${q\left[  X_{j}\right] >0 }$. This ends the proof.
\end{proof}

Combining Theorems \ref{thrm:axiom-prop-1} and \ref{thrm:axiom-prop-2}, we can conclude that if a RS rule has strongly aggregate contribution-over-$q$ ratios with a normalized and additive risk metric, then it must be the $q$-proportional RS rule, and hence it also satisfies the reshuffling property and it has source-anonymous contribution-over-$q$ ratios. On the other hand, taking into account the independence of the axioms 1 and 4, see Proposition~\ref{prop:theorem3}, we know that both properties 1 and 4 have to hold before we can conclude that a RS rule has strongly aggregate contribution-over-$q$ ratios.  

Theorems~\ref{thrm:axiom-prop-1} and~\ref{thrm:axiom-prop-2} provide axiomatizations for the $q$-proportional RS rule. Notice however that compared to Theorem~\ref{thrm:axiom-prop-1}, Theorem~\ref{thrm:axiom-prop-2} is applicable for a smaller class of $q$-proportional RS rules only. In particular, both theorems offer axiomatizations for the mean-proportional RS rule when $q$ is the mean of the random loss under consideration, as well as for the scenario-based proportional RS rule when $q$ refers to the realization of this random loss in a specific scenario. 

Theorem~\ref{thrm:axiom-prop-1} offers an axiomatization of the uniform RS rule by choosing $q$ as a constant mapping. Such a constant mapping obviously is neither normalized nor additive. 
It is important to note that the uniform RS rule does not fall within the class of $q$-proportional RS rules with additive and normalized risk metric $q$. To prove this statement, assume for a moment that the uniform RS rule belongs to this class, meaning that there exists a normalized and additive risk measure $q:\chi\rightarrow \mathbb{R}$ such that for any $\boldsymbol{X}\in\chi^{n}$, we have that
\[
C_{1}^{\text{uni}}\left[  \boldsymbol{X}\right]  =\frac{S_{\boldsymbol{X}}}{n}=\frac{q[X_{1}]}{q[S_{\boldsymbol{X}}]}\times {S_{\boldsymbol{X}}}.
\]
Now, consider the loss vector $\boldsymbol{X}=\big(0,X_{2},\ldots,X_{n}
\big)$. Given the normalization property of $q$, this would imply that
$q[X_{1}]=0$, which leads to the conclusion that $S_{\boldsymbol{X}}=0$. However, this conclusion is generally incorrect, as $S_{\boldsymbol{X}}$ is not necessarily zero. Therefore, we can conclude that the uniform RS rule cannot be classified as a $q$-proportional RS rule with normalized risk metric $q$. We can conclude that the uniform RS rule fits in the setting of Theorem~\ref{thrm:axiom-prop-1}, but not in the one of Theorem~\ref{thrm:axiom-prop-2}.

Examples of $q$-proportional RS rules that rely on non-additive risk metrics $q$ are the variance-proportional RS rule in Example \ref{eq:varRS} and the standard deviation-proportional RS rule in Example \ref{eq:sdRS}.

\section{Characterizing $(q_1,q_2)$-based linear RS rules}\label{section:charac_qqlinear}

Similarly to Section~\ref{section:charac_qprop}, we present two characterizations of the $(q_1,q_2)$-based linear RS rule. These characterizations are also derived by adapting the axioms from Theorem~\ref{thrm:axiom-uni-1} and Theorem~\ref{thrm:axiom-uni-2}, which led to the uniform RS rule.

The $(q_1,q_2)$-based linear RS rule satisfies the reshuffling property, but it satisfies neither the source-anonymous contributions property (Definitions~\ref{def:source_anonymous}) nor the source-anonymous contribution-over-$q$ ratios property (Definition~\ref{def:source_anonymous_ratios}). 
Therefore, we will adapt the two properties to better align with the $(q_1,q_2)$-based linear RS rule. To address this, we introduce the concept of a `$(q_1,q_2)$-standardized contribution' for each participant $i$, which is defined as follows:
\[
\frac{C_i[\mathbf{X}]- q_1[X_{i}]}{q_2[X_{i}, S_{\mathbf{X}}]}.
\]
These standardized contributions allow us to compare the relative size of different contributions, involving distinct risk metrics $q_1$ and $q_2$. Within this framework, $q_1$ serves as the `baseline' or reference value, while $q_2$ serves as the `adjustment' or scaling factor.  This approach enables us to evaluate how each participant's contribution compares relative to the baseline and adjustment factor associated with his loss.

An example is the $q$-proportional RS rule, where both the baseline and the adjustment are equal to $q$. In this case, the `standardized contribution' can be simplified to what we called the `contribution-over-$q$-ratio', see Section~\ref{section:charac_qprop}. Another example is  the covariance-based linear RS rule (Definiton~\ref{def:cov-linear}) which results from assigning $q_1[X_i] = \mathbb{E}[X_i]$ and $q_2[X_i,S_{\mathbf{X}}] = \mathrm{cov}(X_i,S_{\mathbf{X}})$.

Let us now adapt the source-anonymity property by requiring that instead of the contributions of the participants, it is their standardized contributions that remain unchanged after a reshuffle of the individual losses.

\begin{definition}
[Source-anonymous $(q_1,q_2)$-standardized contributions]
Consider the risk metrics $q_1 : \chi \rightarrow \mathbb{R}_0^+$ and $q_2 : \chi^2 \rightarrow \mathbb{R}$. The $(q_1,q_2)$-standardized contributions of a RS rule $\boldsymbol{C}$ are said to be source-anonymous if, for any pool $\boldsymbol{X}$ and any of its reshuffles $\boldsymbol{X}^{\pi}$, the following conditions hold:
\[
{C_i[\mathbf{X}^{\pi}]- q_1[X_{\pi(i)}]}\ = \frac{q_2[X_{\pi(i)}, S_{\mathbf{X}}]}{q_2[X_{i}, S_{\mathbf{X}}]} {(C_i[\mathbf{X}]- q_1[X_{i}])}, \text{ for any }i=1,\ldots,n \text{ with } {q_2[X_{i}, S_{\mathbf{X}}]\neq 0 }. 
\]
\end{definition} 

If a RS rule features source-anonymous $(q_1,q_2)$-standardized contributions, the specific assignment of losses $X_{1}, X_{2}, \ldots, X_{n}$ to individual participants is irrelevant when determining their relative contributions. This property implies that, following a reshuffling of losses, the only necessary information to adjust the original contributions are the $q_1$ and $q_2$ values associated with those losses. The source-anonymity of standardized contributions thus ensures that regardless of which loss a participant incurs, their contribution is adjusted consistently in relation to the loss's baseline $q_1$ and adjustment factor $q_2$.

It is easy to verify that any $(q_1,q_2)$-based linear RS rule has source-anonymous $(q_1,q_2)$-standardized contributions. In particular, the uniform RS rule satisfies this property when both $q_1$ and $q_2$ are constant functions.

Let us now prove a first axiomatization for $(q_1,q_2)$-based linear RS rules.

\begin{theorem}\label{thrm:axiom-gen-1}
Consider the risk metrics $q_1 : \chi \rightarrow \mathbb{R}_0^+$ and $q_2 : \chi^2 \rightarrow \mathbb{R}$. A RS rule $\boldsymbol{C}$ is the $(q_1,q_2)$-based linear RS rule if and only if it satisfies the following two axioms:
\begin{enumerate}
    \item[1.] $\boldsymbol{C}$ satisfies the reshuffling property.
    \item[6.] $\boldsymbol{C}$ has source-anonymous $(q_1,q_2)$-standardized contributions.
\end{enumerate}
\end{theorem}

\begin{proof}
The ``only if'' direction is straightforward. We now proceed to prove the``if'' direction in a similar spirit to the proof of Theorem~\ref{thrm:axiom-prop-1}.  Consider a pool $\boldsymbol{X} \in \chi^{n}$ and its reshuffle $\boldsymbol{X}^{\pi}$. Given that a $(q_1,q_2)$-based linear RS rule is defined only when $\sum_{k=1}^n q_2[X_k, S_{\boldsymbol{X}}]\neq 0$, let us suppose that $q_2[X_j, S_{\boldsymbol{X}}]\neq 0$ for some $j \in \{1,\ldots,n\}$.
Given that the $(q_1,q_2)$-standardized contributions are assumed to be source-anonymous and $\boldsymbol{C}$ satisfies the reshuffling property, we have that
\[
C_{\pi(j)}[\mathbf{X}] - q_1[X_{\pi(j)}] = {C_j[\mathbf{X}^{\pi}]- q_1[X_{\pi(j)}]}\ = \frac{q_2[X_{\pi(j)}, S_{\mathbf{X}}]}{q_2[X_{j}, S_{\mathbf{X}}]} {(C_j[\mathbf{X}]- q_1[X_{j}])}.
\]

Since this expression holds for any permutation $\pi$, it follows that
\[
C_{k}[\mathbf{X}] - q_1[X_{k}] =  \frac{q_2[X_{k}, S_{\mathbf{X}}]}{q_2[X_{j}, S_{\mathbf{X}}]} {(C_j[\mathbf{X}]- q_1[X_{j}])}, \text{\quad for any }k.
\]

Taking into account the the full allocation condition (\ref{eq2}), we find that
\[
\frac{C_j[\mathbf{X}] - q_1[X_j]}{q_2[X_j, S_{\mathbf{X}}]} = \frac{S_{\boldsymbol{X}} - \sum_{i=1}^n q_1[X_i]}{\sum_{i =1}^n q_2[X_i, S_{\boldsymbol{X}}]},
\]
which leads us to the $(q_1,q_2)$-based linear RS rule. This reasoning holds for any $\boldsymbol{X} \in \chi^{n}$ with $q_2[X_j, S_{\boldsymbol{X}}]\neq 0$ for at least one $j\in \{1,..,n\}$. Hence it would hold for any $\boldsymbol{X} \in \chi^{n}$ with $\sum_{k=1}^n q_2[X_k, S_{\boldsymbol{X}}]\neq 0$. This ends the proof.
\end{proof}

Similarly to Proposition~\ref{prop:theorem3}, one can prove that the axioms presented in Theorem~\ref{thrm:axiom-gen-1} are independent.  

\begin{proposition}
The axioms 1 and 6 considered in Theorem~\ref{thrm:axiom-gen-1} are independent.
\end{proposition}

In Sections~\ref{sec:axiom_uni} and~\ref{section:charac_qprop}, we presented axiomatic characterizations based on either the strongly aggregate contributions property (Definition~\ref{def:stong_agg}) or the strongly aggregate contribution-over-$q$ ratios (Definition~\ref{def:strong_agg_ratio}). However, it is clear that neither of the former two properties applies to the $(q_1,q_2)$-based linear RS rule. To account for the unique characteristics of this rule, we now define the following new property.

\begin{definition}
[Strongly aggregate $(q_1,q_2)$-based standardized contributions] Consider the risk metrics $q_1 : \chi \rightarrow \mathbb{R}_0^+$ and $q_2 : \chi^2 \rightarrow \mathbb{R}$. A RS rule $\boldsymbol{C}$  has strongly aggregate $(q_1,q_2)$-based standardized contributions if there exists a function
\[
\mathbf{h} : \mathbb{R}^3 \rightarrow \mathbb{R}^n
\]
such that the relative contributions for any $\boldsymbol{X} \in \chi^n$ with $q_2[X_j, S_{\mathbf{X}}]\neq0$ for at least one $j \in \{1,\ldots,n\}$, can be expressed as
\[
C_i[\mathbf{X}] = q_1[X_i] + q_2[X_i, S_{\mathbf{X}}] \cdot h_i(S_{\mathbf{X}}, q_1[S_{\mathbf{X}}], q_2[S_{\mathbf{X}}, S_{\mathbf{X}}]), \quad \text{for any } i=1,\ldots,n.
\]
\end{definition}

This property implies that the contribution of participant $i$ depends on his individual random loss only via the metrics $q_1[X_i]$ and $q_2[X_i, S_{\mathbf{X}}]$. The effect that does not depend on the individual loss can be expressed through a function specific to individual $i$, depending on the realization of the aggregate loss, and the two metrics calculated on that random aggregate loss.

It follows that if 
both metrics are additive in their first argument, then the $(q_1,q_2)$-based linear RS rule results in strongly aggregate $(q_1,q_2)$-based standardized contributions, whereas the uniform RS rule does not. The covariance-based linear RS rule can be expressed as 
$$ C_i^{\mathrm{cov}}[\mathbf{X}]  = \mathbb{E}[X_i] + \mathrm{cov}(X_i, S_{\mathbf{X}}) \cdot h_i(S_{\mathbf{X}}, \mathbb{E}[S_{\mathbf{X}}], \mathrm{cov}(S_{\mathbf{X}},S_{\mathbf{X}})), \quad \text{for any } i=1,\ldots,n.
$$


We now proceed to establish a second characterization for the generalized RS rule.

\begin{theorem}
Consider the risk metrics $q_1 : \chi \rightarrow \mathbb{R}_0^+$ and $q_2 : \chi^2 \rightarrow \mathbb{R}$. Let $q_1[0]=0$ and $q_2[0, \cdot]=0$, and let both measures be additive in their first argument. Then a RS rule $\boldsymbol{C}$ is the $(q_1,q_2)$-based linear RS rule if and only if it satisfies the following axiom:
\begin{enumerate}
    \item[7.] $\boldsymbol{C}$ has strongly aggregate $(q_1,q_2)$-based standardized contributions.
\end{enumerate}
\end{theorem}

\begin{proof}
The ``only if'' direction is straightforward. Let us now prove the ``if'' direction. Assume that $\boldsymbol{C}$ has strongly aggregate $(q_1,q_2)$-based standardized contributions. Consider the loss vector $\boldsymbol{X}$ with $\sum_{k=1}^n q_2[X_k, S_{\boldsymbol{X}}] \neq 0$. Then, $q_2[S_{\mathbf{X}}, S_{\mathbf{X}}] = \sum_{k=1}^n q_2[X_k, S_{\boldsymbol{X}}] \neq 0$ and 
$q_2[X_j, S_{\mathbf{X}}]\neq 0$ for at least one $j\in \{1,\ldots,n\}$. The proof follows in the same spirit as the proof of Theorem~\ref{thrm:axiom-prop-2}. Since $q_1[0]=0$ and $q_2[0, \cdot]=0$, we have that
\(
C_{i}[\left( S_{\boldsymbol{X}}, 0, \ldots, 0 \right)] = 0
\) for any $i \in \{2, \ldots, n\}.$ Using the full allocation condition (\ref{eq2}), we can conclude that $C_{1}[\left( S_{\boldsymbol{X}}, 0, \ldots, 0 \right)] = S_{\boldsymbol{X}}$. Thus we have that
\begin{align*}
C_{1}[\boldsymbol{X}] &= q_1[X_1] + q_2[X_1, S_{\mathbf{X}}] \cdot h_1(S_{\mathbf{X}}, q_1[S_{\mathbf{X}}], q_2[S_{\mathbf{X}}, S_{\mathbf{X}}]) \nonumber \\
&= q_1[X_1] + q_2[X_1, S_{\mathbf{X}}] \cdot \frac{C_{1}[\left( S_{\boldsymbol{X}}, 0, \ldots, 0 \right)] - q_1[S_{\boldsymbol{X}}]}{q_2[S_{\mathbf{X}}, S_{\mathbf{X}}]}.\\
&=  q_1[X_1] + \frac{q_2[X_1, S_{\boldsymbol{X}}]}{\sum_{i = 1}^n q_2[X_i, S_{\boldsymbol{X}}]} \left(S_{\boldsymbol{X}} - \sum_{i = 1}^n q_1[X_i]\right).
\end{align*}
Proceeding in a similar way for $i = 2, 3, \ldots, n$, we can conclude that $\boldsymbol{C}$ is the $(q_1,q_2)$-based linear RS rule.
\end{proof}


\section{Concluding remarks}



In this paper, we explored axiomatic characterizations of some simple risk-sharing rules, such as the uniform RS rule and the mean-proportional RS rule, in order to better understand the key properties behind these rules. To describe the uniform RS rule, we identified and explained three essential properties:
\begin{enumerate}
    \item Reshuffling: If participants exchange their individual losses, their contributions toward covering those losses will be exchanged accordingly.
    \item Source-anonymous contributions: Each participant's contribution depends only on the total set of individual losses, without regard to which participant 'owns' which loss. 
    \item Strongly aggregate contributions: The randomness of the contributions is caused only by the aggregate loss, and the specific functions $h_i$ that describe how to divide the aggregate loss across participants are the same for all pools.
\end{enumerate}

To the best of our knowledge, the concept of source-anonymous contributions is new to the literature. We demonstrated that the two properties of reshuffling and source-anonymous contributions are satisfied by the uniform RS  rule and moreover, a RS rule that satisfies these two properties can only be the uniform RS rule. Additionally, both reshuffling and strongly aggregate contributions are necessary and sufficient conditions for this rule. We proved that both axiomatic characterizations for the uniform RS rule are non-redundant, in the sense that the underlying axioms are independent.

These elementary axioms for the uniform RS rule served as the foundation for exploring two other broader classes of fundamental RS rules: the $q$-proportional RS rules and $(q_1,q_2)$-based linear RS rules. These include familiar RS rules such as the mean-proportional RS rule and the covariance-based linear RS rule. In order to expand the idea of source-anonymity, we needed to redefine contributions in terms of contribution-over-$q$ ratios for $q$-proportional RS rules, and use standardized expressions for $(q_1,q_2)$-based linear RS rules. Furthermore, the concept of strongly aggregate contributions had to be similarly modified for these new classes of RS rules. These modifications allowed us to establish clear necessary and sufficient conditions for these extended classes of RS rules.

Finally, we introduced scenario-based RS rules, which provide novel examples of the classes of $q$-proportional RS rules and $(q_1,q_2)$-based linear RS rules. Under these new rules, risk-sharing is determined according to predefined 'typical' and 'extreme' scenarios. Scenario-based risk-sharing  does not require knowledge of probability theory, but relies solely on expert judgments or opinions. These scenario-based RS rules may be particularly relevant, for example, in small community-based risk-sharing arrangements, where the lack of probabilistic knowledge may be compensated by common `farmers' sense, which plays an important role in such communities.

\section*{Acknowledgments}
Jan Dhaene and Rodrigue Kazzi gratefully acknowledge funding from FWO and F.R.S.-FNRS under the Excellence of Science (EOS) programme, project ASTeRISK (40007517). The authors also thank Michel Denuit, Runhuan Feng and Mario Ghossoub for helpful comments on an earlier version of the paper. The authors also thank the audience at the UConn Actuarial Science Seminar on 11 November 2024, where this paper was presented.

\bibliographystyle{apalike}
\bibliography{references}

\end{document}